\documentclass[11pt]{article}

\usepackage[T1]{fontenc}

\usepackage{amsmath,amssymb,amsthm}
\usepackage{graphicx}
\usepackage{psfrag}
\usepackage[active]{srcltx}

\newcommand{\R}{\mathbb{R}}

\newcommand{\D}{{\cal D}}
\newcommand{\ot}{\otimes}

\newcommand{\bld}[1]{\boldsymbol{#1}}

\DeclareMathOperator{\eee}{\!\text{\em e}}
\newcommand{\src}{{\sideset{^s}{}{\eee}}}
\newcommand{\tar}{{\sideset{^t}{}{\eee}}}
\DeclareMathOperator{\intr}{{\rm Int}}

\DeclareMathOperator{\sgn}{{\rm sgn}}
\DeclareMathOperator{\supp}{{\rm supp}}
\newcounter{mnotecount}[section]

\newtheorem{thr}{Theorem}
\newtheorem{lm}[thr]{Lemma}

\newtheorem{cor}[thr]{Corollary}

\numberwithin{equation}{section}
\numberwithin{thr}{section}

\begin{document}

\title{Variables suitable for constructing quantum states for the Teleparallel Equivalent of General Relativity I\footnote{This is an author-created version of a paper published as {\em Gen. Rel. Grav.} {\bf 46} 1620 (2014) DOI 10.1007/s10714-013-1620-z.}}
\author{ Andrzej Oko{\l}\'ow}
\date{March 3, 2014}

\maketitle
\begin{center}
{\it  Institute of Theoretical Physics, Warsaw University\\ ul. Ho\.{z}a 69, 00-681 Warsaw, Poland\smallskip\\
oko@fuw.edu.pl}
\end{center}
\medskip

\begin{abstract}
We present the first part of an analysis aimed at introducing variables which are suitable for constructing a space of quantum states for the Teleparallel Equivalent of General Relativity via projective techniques---the space is meant to be applied in a canonical quantization of the theory. We show that natural configuration variables on the phase space of the theory can be used to construct a space of quantum states which however possesses an undesired property. We introduce then a family of new variables such that some elements of the family can be applied to build a space of quantum states free of that property. 
\end{abstract}

\section{Introduction \label{intro}}

A formulation of general relativity called Teleparallel Equivalent of General Relativity (TEGR)\footnote{See \cite{mal-rev} for the newest review on TEGR.} has not been yet used as a starting point for a quantization of gravity \cite{app,carlip}. Since nowadays no existing approach to quantum gravity seems to be fully successful it is worth to check whether it is possible to construct a model of quantum gravity based on TEGR. In this paper we will address an issue of constructing a space of quantum states for TEGR which could be applied in the procedure of canonical (or a canonical-like) quantization of the theory.

A Hamiltonian analysis of TEGR \cite{nester,bl,maluf-1,maluf,oko-tegr} shows that it is a constrained system. Since we do not expect that constraints on the phase space of TEGR can be solved classically we would like to apply the Dirac's approach to canonical quantization of constrained systems. According to this approach one first constructs a space of kinematic quantum states, that is, quantum states which correspond to classical states constituting the unconstrained phase space, next among kinematic quantum states one distinguishes physical quantum states as those corresponding to classical states which satisfy all constraints. Thus our goal is to construct a space of kinematic quantum states for TEGR. 

Since TEGR is a background independent theory it is desirable to construct a space of quantum states for it in a background independent manner. Methods which provide a construction of this sort are known from Loop Quantum Gravity (LQG)---see e.g. \cite{rev,rev-1} and references therein---but because of a reason explained below they are rather not applicable to TEGR. Therefore we are going to construct the desired space for TEGR by means of a general method \cite{q-nl} deliberately developed for this purpose. This method works as follows. 

The starting point for the method is a phase space of a theory of the form $P\times \Theta$, where $P$ is a space of momenta, and $\Theta$ is a (Hamiltonian) configuration space (that is, a space of ``positions''). One starts the construction by choosing a set $\cal K$ of real functions on {$\Theta$} called {\em configurational elementary degrees of freedom}. Analogously, one chooses a set of {\em momentum elementary degrees of freedom} consisting of some real functions on $P$. Next, one defines a special directed set $(\Lambda,\geq)$---each element of this set corresponds to a finite collection of both configurational and momentum elementary d.o.f.---and with every element $\lambda$ of $\Lambda$ one associates a set of quantum states denoted by $\D_\lambda$.   

Given $\lambda\in\Lambda$, the set $\D_\lambda$ of quantum stated is constructed as follows. The element $\lambda$ corresponds to a finite set $K$ of configurational d.o.f.. One uses the d.o.f. in $K$ to reduce ``infinite-dimensional'' space $\Theta$ to a finite dimensional space $\Theta_K$---this reduction consists in identifying all points of $\Theta$ for which each d.o.f. in $K$ gives the same value. Then one defines a Hilbert space of functions on $\Theta_K$ square integrable with respect to a measure on $\Theta_K$. The set $\D_\lambda$ is a set of all density operators (i.e. positive operators of trace equal $1$) on this Hilbert space---because density operators represent some (mixed, in general,) quantum states one can treat $\D_\lambda$ as a set of such states.  

In this way one obtains a family $\{\D_\lambda\}_{\lambda\in\Lambda}$ of sets of quantum states. If the set $(\Lambda,\geq)$ is chosen properly then it naturally generates on $\{\D_\lambda\}_{\lambda\in\Lambda}$  the structure of a projective family. Finally, the desired space of kinematic quantum states related to the original phase space $P\times \Theta$ is defined as the projective limit of the family.               

As shown in \cite{q-nl}, the task of constructing such a space of quantum states reduces to a construction of a directed set $(\Lambda,\geq)$ satisfying some assumptions---these assumptions are imposed both on elementary d.o.f. constituting elements of $\Lambda$ and the relation $\geq$. Since now a directed set $(\Lambda,\geq)$ satisfying all these assumption will be called {\em proper directed set $(\Lambda,\geq)$.}

The goal of the present paper is to find variables on the (Hamiltonian) configuration space $\Theta$ of TEGR which are suitable for constructing a proper directed set $(\Lambda,\geq)$ for the theory. More precisely, we are looking for variables on the configuration space which provide a set $\cal K$ of configurational d.o.f. such that 
\begin{enumerate}
\item d.o.f. in $\cal K$ separate points of $\Theta$; \label{as-sep}
\item d.o.f. in $\cal K$ are defined via integrals of functions of components of the variables; the functions are polynomials of the components of degree $1$; \label{as-poly}  \item there exists a directed set elements of which are finite subsets of $\cal K$ such that for every element $K$ of the directed set there exists a natural bijection from $\Theta_K$ onto $\R^N$, where $N$ is the number of d.o.f. in $K$; \label{as-bij}
\item d.o.f. in $\cal K$ are defined in a background independent way  i.e. without application of any background field. \label{as-bgi}
\end{enumerate}

The first three {\bf Assumptions} above correspond to some assumptions imposed in \cite{q-nl} on a proper set $(\Lambda,\geq)$. The present Assumption \ref{as-sep} can be found in Section 2 of \cite{q-nl} containing preliminaries and the Assumption \ref{as-poly} above describes a practical way to satisfy Assumption 3b of \cite{q-nl} (see Section 3.2 and Section 6.2 in that paper). The present Assumption \ref{as-bij} corresponds to Assumption 2 of \cite{q-nl} (see Section 3.2 in that paper). Let us note that the original Assumption 2 is imposed on every finite subset $K$ of $\cal K$ which (together with a finite set $\hat{F}$ of momentum d.o.f.) constitute an element of $(\Lambda,\geq)$: ``if $(\hat{F},K)\in\Lambda$, then...''. But we do not have any set $(\Lambda,\geq)$ for TEGR yet---we are at a stage of preparations for constructing such a set---and therefore we cannot impose the original Assumption 2 as it is formulated in \cite{q-nl}. Instead, we require the existence of a directed set consisting of some special finite subsets of $\cal K$---formulating in this way the present Assumption \ref{as-bij} we hope that a directed set of this sort may facilitate a construction of a proper directed set $(\Lambda,\geq)$ for TEGR. Finally, Assumption \ref{as-bgi} express our wish to construct quantum states for TEGR in a background independent manner.

Results of our inquiries can be summarized as follows: we will find two kinds of variables on the configurations space $\Theta$ of TEGR which not only satisfy the four assumptions above but can be actually used in a background independent manner to construct two distinct spaces of quantum states for TEGR. One of these variables are natural configurational variables on the phase space of TEGR, that is, one-forms $(\theta^A)$, $A=0,1,2,3$, defined on a three-dimensional manifold being a space-like slice of a spacetime. We will show, however, that the space of quantum states derived from these variables possesses an undesired property. Therefore we will transform the natural variables obtaining a family of new variables such that some elements of the family can be used to build a space of quantum states for TEGR free of that property---a construction of this space will be presented in \cite{q-stat}.     

Let us emphasize that the analysis of variables suitable for constructing a space of quantum states for TEGR will be continued in an accompanying paper \cite{ham-nv} where we will analyze more closely the family of new variables. 

Some constructions presented in the present paper are similar to (elements of) a construction of a space of kinematic quantum states for a simple background independent theory called Degenerate Pleba\'nski Gravity (DPG)---the latter construction is described in \cite{q-nl}. It seems to us that it may be quite helpful for the reader to study first the construction in \cite{q-nl} since it is simpler that ones described here.

Let us finally explain why the LQG methods of constructing quantum states do not seem to be applicable to TEGR. The reason is quite simple: the methods require finite dimensional spaces {$\{\Theta_K\}$} to be {\em compact}\footnote{See \cite{non-comp} for a discussion of obstacles which appear if one tries to apply the LQG methods for non-compact spaces {$\{\Theta_K\}$}.} and it is rather difficult to obtain naturally such spaces in the case of TEGR.  

The paper is organized as follows: Section 2 contains preliminaries, in Section 3 we consider the natural variables $(\theta^A)$ and explain why the space of quantum state constructed from them  does not seem to be very promising for canonical quantization of TEGR. In Section 4 we present the family of new variables. Section 5 contains a short summary and an outline of the analysis to be presented in the accompanying paper \cite{ham-nv}. In Appendix we prove two very important lemmas which guarantee that both kinds of variables considered in this paper provide d.o.f. satisfying Assumption \ref{as-bij} above.   

\section{Preliminaries}

\subsection{Vector spaces with scalar products}

Let $\mathbb{M}$ be a four-dimensional oriented vector space equipped with a scalar product $\eta$ of signature $(-,+,+,+)$. We fix an orthonormal basis $(v_A)$ $(A=0,1,2,3)$ of $\mathbb{M}$ such that the components $(\eta_{AB})$ of $\eta$ given by the basis form the matrix ${\rm diag}(-1,1,1,1)$. The matrix $(\eta_{AB})$ and its inverse $(\eta^{AB})$ will be used to, respectively, lower and raise capital Latin letter indices $A,B,C,D\in\{0,1,2,3\}$. 

Denote by $\mathbb{E}$ the subspace of $\mathbb{M}$ spanned by the vectors $\{v_1,v_2,v_3\}$. The scalar product $\eta$ induces on $\mathbb{E}$ a positive definite scalar product $\delta$---its components $(\delta_{IJ})$ in the basis $(v_1,v_2,v_3)$ form a matrix ${\rm diag}(1,1,1)$. The matrix $(\delta_{IJ})$ and its inverse $(\delta^{IJ})$ will be used to, respectively, lower and raise capital Latin letter indices $I,J,K,L,M\in\{1,2,3\}$.

\subsection{Phase space \label{Theta-sec}}

In this paper we will consider a particular phase space being a set of some fields on a three-dimensional oriented connected smooth\footnote{Throughout the paper ``smooth'' means ``of $C^\infty$ class''.} manifold $\Sigma$. A point in the phase space consists of: 
\begin{enumerate}
\item a quadruplet of smooth one-forms $(\theta^{A})\equiv\theta$ on $\Sigma$ such that\footnote{Conditions \ref{lin-cond} and \ref{q-cond} are not independent---in fact, the former is implied by the latter \cite{ham-nv}, but for further considerations it will be convenient to formulate them separately.} 
\begin{enumerate}
\item at each point {$y\in\Sigma$} {three of four one-forms $(\theta^A(y))$ are linearly independent},\label{lin-cond}
\item the metric 
\begin{equation}
q=\eta_{AB}\theta^A\ot\theta^B
\label{q}
\end{equation}
on $\Sigma$ is Riemannian (positive definite). \label{q-cond} 
\end{enumerate}
\item a quadruplet of smooth two-forms $(p_A)$ on $\Sigma$; $p_A$ is the momentum conjugate to $\theta^A$. 
\end{enumerate}         
Since now $\Theta$ will denote the space of all quadruplets $(\theta^A)$ satisfying the  {\em Conditions} above and $P$ will denote the space of all momenta $(p_A)$.  We will call the space $\Theta$ {\em (Hamiltonian) configuration space}. 

The phase space under consideration is then a Cartesian product $P\times \Theta$. As shown in \cite{oko-tegr} and \cite{os}  this is a phase space of both TEGR and a simple theory of the teleparallel geometry called Yang-Mills-type Teleparallel Model\footnote{In \cite{os} while describing the phase space of YMTM we imposed only the weaker and insufficient Condition \ref{lin-cond} and overlooked Condition \ref{q-cond}.} (YMTM) \cite{itin}. 

\subsection{Reduced configuration spaces}

As mentioned above we are going to construct quantum states for TEGR by means of the method described in \cite{q-nl}. Let us recall some notions used in that paper. 

Suppose that a set $\cal K$ of configurational elementary d.o.f. on $\Theta$ is chosen. Given finite set $K=\{\kappa_1,\ldots,\kappa_N\}\subset {\cal K}$ we say that $\theta\in \Theta$ is $K$-equivalent to $\theta'\in\Theta$, 
\[
\theta\sim_K \theta',
\]
if for every $\kappa_I\in K$
\[
\kappa_I(\theta)=\kappa_I(\theta').
\]      
The relation $\sim_K$ is an equivalence one and therefore it defines a quotient space
\[
\Theta_K:=\Theta/\sim_K.
\] 
We will denote by $[\theta]$ an equivalence class given by $\theta$. 

There exists a natural\footnote{The set $K$ is unordered, thus to define the map $\tilde{K}$ one has to order elements of $K$. Thus the map $\tilde{K}$ is natural modulo the ordering. However, every choice of the ordering is equally well suited for our purposes and nothing essential depends on the choice. Therefore we will neglect this subtlety throughout the paper.} injective map from {$\Theta_K$} into $\R^N$:  
\begin{equation}
\Theta_K\ni[\theta]\mapsto\tilde{K}([\theta]):=\Big(\kappa_1(\theta),\ldots,\kappa_N(\theta)\Big)\in \R^N.
\label{k-inj}
\end{equation}

We will say that the d.o.f. in $K$ are {\em independent} if the image of $\tilde{K}$ is an $N$-dimensional submanifold of $\R^N$. The set $\Theta_K$ given by a set $K$ of independent d.o.f. will be called a {\em reduced configuration space}. 

Let us note that the formulation of Assumption \ref{as-bij} in Section \ref{intro} lacks some precision since there we did not define what the ``natural bijection from {$\Theta_K$} onto $\R^N$'' is. Now we can formulate the assumption strictly:
\begin{enumerate}
\setcounter{enumi}{2}  
\item there exists a directed set elements of which are finite subsets of $\cal K$ such that for every element $K$ of the directed set the map $\tilde{K}$ given by \eqref{k-inj} is a bijection or, equivalently, 
\begin{equation}
\Theta_K\cong \R^N
\label{As}
\end{equation}
under $\tilde{K}$, where $N$ is the number of elements of $K$. 
\end{enumerate}
  
\section{Natural variables on $\Theta$}

\subsection{Configurational elementary d.o.f.}

Let us use the natural\footnote{The variables are natural in this sense that they are a result of the Legendre transformation \cite{nester,maluf,oko-tegr} applied to a Lagrangian formulation of TEGR in terms of cotetrad fields on a four-dimensional manifold.} variables $(\theta^A)$ on $\Theta$ to define configurational elementary d.o.f.. Since the variables are one-forms we follow the LQG methods (see \cite{rev,rev-1}) and define the following real function on $\Theta$:
\begin{equation}
\Theta\ni\theta \mapsto \kappa^A_e(\theta):=\int_e\theta^A\in\R,
\label{edf-0}
\end{equation}
 where $e$ is an {\em edge}\footnote{A {\em simple edge} is a one-dimensional connected $C^\infty$ submanifold of $\Sigma$ with two-point boundary. An edge is an {\em oriented} one-dimensional connected $C^0$ submanifold of $\Sigma$ given by a finite union of simple edges.} in $\Sigma$. Let 
\[
\bar{\cal K}:=\{\ \kappa^A_e \ \},
\]  
where $A=0,1,2,3$ and $e$ runs over a set of all edges in $\Sigma$. We choose $\bar{\cal K}$ to be a set of configurational elementary d.o.f. generated by the natural variables.

Now we have to check whether the set $\bar{\cal K}$ satisfies Assumptions listed in Section \ref{intro}. It is clear that functions in $\bar{\cal K}$ separate points of $\Theta$, thus $\bar{\cal K}$ meets Assumption \ref{as-sep}. The function $\kappa^A_e(\theta)$ can be easily expressed in terms of components of the one-form $\theta^A$ given by local coordinate frames on $\Sigma$. It follows immediately from such expressions that $\bar{\cal K}$ satisfies Assumption \ref{as-poly}.                

Regarding Assumption \ref{as-bij}, let us focus on sets of d.o.f. given by {\em graphs}\footnote{We say that two edges are {\em independent} if the set of their common points is either empty or consist of one or two endpoints of the edges. A {\em graph} in $\Sigma$ is a finite set of pairwise independent edges.} in $\Sigma$---it is known from LQG that under a technical requirement\footnote{One assumes $\Sigma$ to be a real-analytic manifold and restrict oneself to edges built from {\em analytic} simple edges.} all graphs in $\Sigma$ form a directed set. Consider then a graph $\gamma$ being a collection $\{e_1,\ldots,e_N\}$ of edges in $\Sigma$. The graph defines a finite set 
\[
K_\gamma:=\{\ \kappa^A_{e_1},\ldots,\kappa^A_{e_N}\ | \ A=0,1,2,3 \ \}
\]
of elementary d.o.f.. The set $(K_\gamma,\geq)$, where $\gamma$ runs over the directed set of graphs in $\Sigma$ and the relation $\geq$ is induced by the directing relation on the set of graphs, is a directed set.    

There holds the following lemma proven in Appendix \ref{theta-x-prf}:
\begin{lm}
Let $\gamma=\{e_1,\ldots,e_N\}$ be a graph. Then for every $(x^A_{j})\in\R^{4N}$ there exists $\theta\in\Theta$ such that
\[
\kappa^A_{e_{j}}(\theta)=x^A_{j}
\]      
for every $A=0,1,2,3$ and $j=1,2,\ldots,N$. 
\label{theta-x} 
\end{lm}   
    
Let us now comment on the lemma. Recall now that Condition \ref{q-cond} of the phase space description presented in Section \ref{Theta-sec} means that for every $\theta\in\Theta$ and for every nonzero vector $X$ tangent to $\Sigma$ the values $(\theta^A(X))$ form a space-like vector in $\mathbb{M}\cong\R^4$. On the other hand, given edge $e$ and $\theta\in\Theta$, we can interpret a quadruplet $(\kappa^A_e(\theta))$ as a vector in $\mathbb{M}$. Naively thinking, one could expect that $(\kappa^A_e(\theta))$ should be space-like also. However, a sum---and then an integral---of space-like vectors in $\mathbb{M}$ may be any other vector in $\mathbb{M}$ and this is exactly why the lemma is true.    

Note that Lemma \ref{theta-x} implies that for every graph $\gamma$ in $\Sigma$ the map $\tilde{K}_\gamma$ (see \eqref{k-inj}) is a bijection or, equivalently, the reduced configuration space 
\[
\Theta_{K_\gamma}\cong\R^{4N},
\]
where $N$ is the number of edges of $\gamma$. Consequently, the set $\bar{\cal K}$ with the directed set $(K_\gamma,\geq)$ satisfies Assumption \ref{as-bij}. 

It is clear that the d.o.f. in $\bar{\cal K}$ are defined in a background independent manner. Note that there exists on $\bar{\cal K}$ a natural action of diffeomorphisms on $\Sigma$: given diffeomorphism $\varphi$ on $\Sigma$, a d.o.f. $\kappa^A_e\in\bar{\cal K}$ is mapped by the diffeomorphism to $\varphi^*\kappa^A_e$ being a function on $\Theta$ such that
\[
(\varphi^*\kappa^A_e)(\theta):=\int_e\varphi^*\theta^A=\int_{\varphi(e)}\theta^A=\kappa^A_{\varphi(e)}(\theta).
\]     
This means that $\bar{\cal K}$ is preserved by the action of the diffeomorphisms. 

We conclude that the set $\bar{\cal K}$ of configurational d.o.f. defined by the natural variables $(\theta^A)$ satisfies all Assumptions presented in Section \ref{intro}. Thus the set $\bar{\cal K}$ seems to be suitable for constructing in a background independent way a set of quantum states for TEGR. In fact, the directed set $(K_\gamma,\geq)$ can be extended to a proper directed set $({\Lambda},\geq)$ for TEGR---the construction of the latter set is fully analogous to the construction of a set $(\Lambda,\geq)$ for DPG \cite{q-nl}. Since the resulting set $({\Lambda},\geq)$ for TEGR is proper it generates {\em a space of kinematic quantum states} for TEGR which will be denoted by $\bar{\D}$.     

\subsection{An undesired property $\bar{\D}$ }

Unfortunately, the space $\bar{\D}$ of kinematic quantum states for TEGR seems to be too large to be used in a canonical quantization of TEGR. The space is ``too large'' in the following sense.

Let us denote by $\Theta'$  the set of all quadruplets $(\theta^A)$ of one-forms on $\Sigma$ which satisfy Condition \ref{lin-cond} of the phase space description (see Section \ref{Theta-sec}). Obviously, $\Theta\subset\Theta'$ and consequently Lemma \ref{theta-x} is true in the case of $\Theta'$. Defining the space $\Theta'_{K_\gamma}$ analogously to $\Theta_{K_\gamma}$ and $\D'_{{\lambda}}$ analogously to $\D_{{\lambda}}$ we see immediately that
\begin{align}
\Theta'_{K_\gamma}&\cong\Theta_{K_\gamma},& \D'_{{\lambda}}&\cong\D_{{\lambda}}.
\label{T'=T}
\end{align}
Thus the space $\bar{\D}$ meant to correspond to $\Theta$ corresponds actually to the larger space $\Theta'$.

Note that the space $\Theta'$ contains quadruplets $(\theta^A)$ which via the formula \eqref{q} define on $\Sigma$ not only Riemannian metrics but also metrics which (locally or globally) are Lorentzian (i.e. of signature $(-,+,+)$). Thus the kinematic quantum states in $\bar{\D}$ correspond also to a large set of quadruplets $(\theta^A)$ which have nothing to do with elements of $\Theta$---note that it is rather not possible for a quadruplet defining a Lorentzian metric to be a limit of a sequence of elements of $\Theta$. 

Is it possible to isolate quantum states in $\bar{\D}$ which do not correspond to Loren\-tzian metrics on $\Sigma$? Perhaps it is, but we expect this to be rather difficult because of the following reason. By means of d.o.f. belonging to a finite subset $K_\gamma$ of $\bar{\cal K}$ we are not able to distinguish between elements of $\Theta$ and those of $\Theta'\setminus\Theta$---see the first equation \eqref{T'=T}. On the other hand, all d.o.f. in $\bar{\cal K}$ separate points not only in $\Theta$ but also in $\Theta'$. Thus the all d.o.f. in $\bar{\cal K}$ distinguish between elements of $\Theta$ and $\Theta'\setminus\Theta$. Consequently, we are not able to isolate quantum states which do not correspond to Loren\-tzian metrics by means of a family $\{R_\lambda\}_{\lambda\in\Lambda}$ of restrictions such that each restriction $R_\lambda$ is imposed on elements of $\D_{{\lambda}}$ but would have to isolate desired states at the level of the whole $\bar{\D}$. Taking into account the complexity of $\bar{\D}$, this task seems to be very difficult. Therefore we prefer to find other variables which could give us a space of quantum states free of the undesired property of $\bar{\D}$. 
  
\section{New variables on $\Theta$}

\subsection{New variables---preliminary considerations}

The undesired property of $\bar{\D}$ just described follows from the fact that the variables $(\theta^A)$ can be used to parameterize not only the configuration space $\Theta$ but also the larger space $\Theta'$ (provided Condition \ref{q-cond} has been omitted). Thus to obtain a space of kinematic quantum states for TEGR free of the property of $\bar{\D}$ we can try to find new variables which parameterize the space $\Theta$ and cannot be used to describe those elements of $\Theta'\setminus\Theta$ which correspond to Lorentzian metrics on $\Sigma$. Below we present some preliminary considerations results of which will be used in the next subsection to define such new variables. 
 
Condition \ref{lin-cond} of the phase space description together with continuity of the fields mean that three of four one-forms $(\theta^A)$ define a local coframe on $\Sigma$ and consequently the remaining one-form can be expressed as a linear combination of the three ones. It turns out that Condition \ref{q-cond} allows to formulate a stronger statement:
\begin{lm}
A quadruplet $(\theta^A)$ belongs to $\Theta$ if and only if for every point $y\in\Sigma$   
\begin{enumerate}
\item the forms $(\theta^1(y),\theta^2(y),\theta^3(y))$ are linearly independent, 
\item 
\begin{equation}
\theta^0(y)=\alpha_I(y)\theta^I(y)
\label{th0-ath}
\end{equation}
where $\alpha_I(y)$ are real numbers satisfying 
\begin{equation}
\alpha_I(y)\alpha^I(y)<1.
\label{aa<1}
\end{equation}
\end{enumerate}   
\label{al-th}
\end{lm} 
  
\begin{proof}
Let us fix a point $y\in\Sigma$. For the sake of simplicity till the end of this proof we will omit the symbol ``$y$'' in the notation i.e. we will denote $\theta^A(y)$ by $\theta^A$, $\alpha_I(y)$ by $\alpha_I$ and $q(y)$ by $q$. As before we will refer to the two conditions imposed on the elements of $\Theta$ in Section \ref{Theta-sec} as to, respectively, Condition \ref{lin-cond} and Condition \ref{q-cond} and to the two assertions of the lemma as, respectively, Assertion 1 and Assertion 2.       

\paragraph{Step 1: Conditions \ref{lin-cond} and \ref{q-cond} imply Assertion 1} Condition \ref{lin-cond} means that either $(i)$ $(\theta^1,\theta^2,\theta^3)$  or $(ii)$ $(\theta^0,\theta^I,\theta^J)$, $I\neq J$, are linearly independent. Let us show that $(i)$ is true even if $(ii)$ holds. Without loss of generality we assume that $(\theta^0,\theta^1,\theta^2)$ are linearly independent. Then for some real numbers $a,b,c$ 
\[
\theta^3=a\theta^0+b\theta^1+c\theta^2
\]        
and
\[
q=-\theta^0\ot\theta^0+\theta^1\ot\theta^1+\theta^2\ot\theta^2+(a\theta^0+b\theta^1+c\theta^2)\ot(a\theta^0+b\theta^1+c\theta^2).
\]

Let $Y$ be a vector belonging to  $T_y\Sigma$ such that $\theta^0(Y)=1$ and $\theta^1(Y)=\theta^2(Y)=0$. Because of Condition \ref{q-cond} the number $q(Y,Y)$ must be positive:
\[
q(Y,Y)=-1+a^2>0
\]    
which means that
\begin{equation}
a^2>1.
\label{a^2}
\end{equation}
Now by virtue of \eqref{a^2} and the following equations 
\[
\begin{matrix}
\theta^3&=&a\theta^0&+&b\theta^1&+&c\theta^2\\
\theta^1&=&         & &\theta^1 & &\\
\theta^2&=&         & &         & & \theta^2
\end{matrix}
\]
the forms $(\theta^1,\theta^2,\theta^3)$ are linearly independent. 

\paragraph{Step 2: Condition \ref{q-cond} and Assertion 1 are equivalent to Assertions 1 and 2} If Assertion 1 is true then there exists real numbers $(\alpha_I)$, $I=1,2,3$, such that
\[
\theta^0=\alpha_I\theta^I.
\] 
Consequently
\begin{equation}
q=(-\alpha_I\alpha_J+\delta_{IJ})\theta^I\ot\theta^J.
\label{q-th^I}
\end{equation}
The metric $q$ is positive definite if and only if the eigenvalues of the matrix
\begin{equation}
(q_{IJ}):=(-\alpha_I\alpha_J+\delta_{IJ})
\label{qIJ}
\end{equation}
are positive. Of course, if all the $(\alpha_I)$ are zero then the eigenvalues are positive. Assume then that 
\begin{equation}
\alpha_I\alpha^I>0.
\label{aa>0}
\end{equation}
Then the eigenvectors of the matrix $(q_{IJ})$ are $(\alpha_I)$ and $(\beta_I),(\gamma_I)$, where the latter two vectors satisfy
\[
\sum_{I}\beta_I\alpha_I=\sum_I\gamma_I\alpha_I=\sum_{I}\beta_I\gamma_I=0.
\]   
Indeed,
\[
\sum_Jq_{IJ}\alpha_J=\alpha_I(-\sum_J\alpha_J\alpha_J+1)=(1-\alpha_J\alpha^J)\alpha_I
\]
and
\[
\sum_Jq_{IJ}\beta_J=\beta_I, \ \ \ \sum_Jq_{IJ}\gamma_J=\gamma_I. 
\]
These results mean that the eigenvalues of $(q_{IJ})$ are $1$, $1$  and 
\begin{equation}
1-\alpha_I\alpha^I.
\label{1-aa}
\end{equation}

The conclusion is that $q$ is positive definite if and only if $(i)$ all the $\{\alpha_I\}$ are zero or $(ii)$  $1-\alpha_I\alpha^I>0$ if \eqref{aa>0} holds. Obviously, the alternative of the conditions $(i)$ and $(ii)$ can be equivalently expressed as the following one condition
\[
1-\alpha_I\alpha^I>0.
\]
Thus we showed that Condition \ref{q-cond} and Assertion 1 are equivalent to Assertions 1 and 2. 

\paragraph{Step 3: final conclusion} Clearly, Assertion 1 implies Condition \ref{lin-cond}. This fact together with the result of Step 1 ensure that Conditions \ref{lin-cond} and \ref{q-cond} are equivalent to Condition \ref{q-cond} and Assertion 1. Now to finish the proof it is enough to take into account the result of Step 2.  
\end{proof}

\begin{cor}
If $(\theta^A)\in \Theta$ then the triplet $(\theta^1,\theta^2,\theta^3)$ is a global coframe on $\Sigma$.   
\end{cor}

\begin{proof}
The corollary follows immediately from Assertion 1 of Lemma \ref{al-th}.
\end{proof}

A consequence of the corollary is that the space $\Theta$ splits into two disjoint subspaces: 
\begin{align*}
&\Theta=\Theta_+\cup\Theta_-, & &\Theta_+\cap\Theta_-=\varnothing, 
\end{align*}
where $\Theta_+$ is constituted by quadruplets $(\theta^0,\theta^I)$ such that the coframe $(\theta^I)$ is compatible with the fixed orientation of $\Sigma$ and $\Theta_-$ consists of quadruplets such that $(\theta^I)$ defines the opposite orientation on the manifold.     

Let us finally reformulate Lemma \ref{al-th} in the following way:

\begin{lm}
There exists a one-to-one correspondence between elements of $\Theta$ and all pairs $(\alpha_I,\theta^J)$ consisting of
\begin{enumerate}
\item real functions $\alpha_I$, $I=1,2,3$,  on $\Sigma$ such that 
\begin{equation}
\alpha_I\alpha^I<1,
\label{aa<1'}
\end{equation}
\item one-forms  $\theta^J$, $J=1,2,3$, on $\Sigma$  constituting a global coframe on the manifold. 
\end{enumerate}  
The correspondence is given by 
\begin{equation}
(\alpha_I,\theta^J)\mapsto (\theta^0=\alpha_I\theta^I,\theta^J)\in\Theta.
\label{alth->}
\end{equation}
\label{th-bij-0}
\end{lm}

Note that a collection $(\alpha_I)$ can be treated as a function on $\Sigma$ valued in a unit open ball
\[
\mathbb{B}:=\{\ (a,b,c)\in \R^3 \ | \ a^2+b^2+c^2<1 \ \}.
\]   

Lemma \ref{th-bij-0} guarantees that the space $\Theta$ can be parameterized by global coframes on $\Sigma$ and functions $(\alpha_I)$ on the manifold valued in the ball $\mathbb{B}$. Let us now use these variables to define elementary d.o.f..

Since $(\alpha_I)$ are real functions on $\Sigma$, that is, zero-forms it is natural to use a point $y\in\Sigma$ to define a map 
\begin{equation}
\Theta\ni\theta\mapsto\kappa^{\prime I}_y(\theta):=\alpha^I(y)\in\R.
\label{al-dof}
\end{equation}
On the other hand elementary d.o.f. corresponding to the global coframes can be chosen as before, i.e.,
\begin{equation}
\Theta\ni\theta\mapsto \kappa^I_e(\theta)=\int_e \theta^I\in\R.
\label{th-dof}
\end{equation}
Let 
\[
\bar{\cal K}':=\{\ \kappa^{\prime I}_y, \kappa^J_e \ \},
\]  
where $I,J=1,2,3$, $y$ runs over $\Sigma$  and $e$ over the set of all edges in the manifold. We choose $\bar{\cal K}'$ to be a set of configurational elementary d.o.f. generated by the variables $(\alpha_I,\theta^J)$.

Let us check whether the set $\bar{\cal K}'$ satisfies all Assumptions presented in Section \ref{intro}. It obviously meets Assumption \ref{as-sep}. Note that the r.h.s. of \eqref{al-dof} can be treated as an integral of the function $\alpha^I$ over the set $\{y\}\subset\Sigma$ and, consequently, $\bar{\cal K}'$ satisfies Assumption \ref{as-poly}.    

Regarding Assumption \ref{as-bij}, consider a finite set $u=\{y_1,\ldots,y_M\}$ of points in $\Sigma$ and a graph $\gamma=\{e_1,\ldots,e_N\}$ in the manifold and define a finite set of d.o.f.
\begin{equation}
K'_{u,\gamma}:=\{\ \kappa^{\prime I}_{y_1},\ldots,\kappa^{\prime I}_{y_M},\kappa^J_{e_1},\ldots,\kappa^J_{e_N}\ | \ I,J=1,2,3 \ \}.
\label{K-ug}
\end{equation}
Note that a collection of all such sets is a directed set: we say that $K'_{u',\gamma'}$ is greater than $K'_{u,\gamma}$, 
\[
K'_{u',\gamma'}\geq K'_{u,\gamma},
\]
if $u'\supset u$ and $\gamma'\geq\gamma$. 

Now we have to find the image of the map $\tilde{K'}_{u,\gamma}$ (see \eqref{k-inj}). It is obvious that there holds the following lemma
\begin{lm}
Let $u=\{y_1,\ldots,y_M\}$ be a finite collection of points in $\Sigma$. Then for every $(z^I_{j})\in \mathbb{B}^{M}$ there exist real functions $(\alpha_I)$ satisfying the condition described in Lemma \ref{th-bij-0} such that
\[
\alpha^I(y_{j})=z^I_{j}
\]      
for every $I=1,2,3$ and $j=1,2,\ldots,M$. 
\label{al-b}
\end{lm}
\noindent The next lemma is proven in Appendix \ref{theta-x3-prf}:
\begin{lm}
Let $\gamma=\{e_1,\ldots,e_N\}$ be a graph. Then for every $(x^I_{j})\in\R^{3N}$ there exists a global coframe $(\theta^I)$ on $\Sigma$ compatible (incompatible) with the orientation of the manifold such that  
\[
\int_{e_{j}}\theta^I=x^I_{j}
\]      
for every $I=1,2,3$ and $j=1,2,\ldots,N$. 
\label{theta-x3}
\end{lm}
\noindent The following conclusion is a simple consequence of Lemmas \ref{th-bij-0}, \ref{al-b}, and \ref{theta-x3}:
\begin{cor}
Let $u=\{y_1,\ldots,y_M\}$ be a finite collection of points in $\Sigma$ and $\gamma=\{e_1,\ldots,e_N\}$ be a graph such that either $u$ or $\gamma$ is not an empty set ($N,M\geq0$ but $N+M>0$). Then for every $(z^I_{i},x^J_{j})\in\mathbb{B}^{M}\times \R^{3N}$ there exists $\theta\in\Theta_+(\Theta_-)$ such that  
\begin{align*}
&\kappa^{\prime I}_{y_{i}}(\theta)=z^I_{i},&&\kappa^J_{e_{j}}(\theta)=x^J_{j}
\end{align*}
for every $I,J=1,2,3$, $i=1,\ldots,M$ and $j=1,2,\ldots,N$. 
\label{cor-Kug}
\end{cor}  
\noindent The image of $\tilde{K'}_{u,\gamma}$ is then $\mathbb{B}^M\times\R^{3N}$ and   the reduced configuration space
\[
\Theta_{K'_{u,\gamma}}\cong \mathbb{B}^M\times\R^{3N}.
\]
This means that the set $\bar{\cal K}'$ {\em does not} satisfy Assumption \ref{as-bij}.

On the other hand, $\bar{\cal K}'$ meets Assumption \ref{as-bgi}. A diffeomorphism $\varphi$ on $\Sigma$ maps a d.o.f. $\kappa^{\prime I}_y$ into a function on $\Theta$ given by 
\[
(\varphi^*\kappa^{\prime I}_y)(\theta):=(\varphi^*\alpha^I)(y)=\alpha^I(\varphi(y))=\kappa^{\prime I}_{\varphi(y)}(\theta).
\]       
Of course, $\bar{\cal K}'$ is preserved by the action of all diffeomorphisms on $\Sigma$. 

The conclusion is that the set $\bar{\cal K}'$ does not meet Assumptions \ref{as-bij} but satisfies all remaining ones. Moreover, the variables $(\alpha_I,\theta^J)$ can be used to define Lorentzian metrics on $\Sigma$ provided we give up the condition \eqref{aa<1'}---if $\alpha_I\alpha^I>1$ then the eigenvalue \eqref{1-aa} of the matrix \eqref{qIJ} is negative and the resulting metric \eqref{q-th^I} is Lorentzian. However, there is a progress with respect to the previously considered variables $(\theta^A)$ and the corresponding d.o.f. in $\bar{\cal K}$, because now if a sextuplet $(\alpha_I,\theta^J)$ defines a metric which is Lorentzian on a subset of $\Sigma$ then any triplet $\{ \ \kappa^{\prime I}_y\ | \ I=1,2,3\ \}$ of d.o.f. with $y$ belonging to the subset can be used to distinguish between this sextuplet $(\alpha_I,\theta^J)$ and ones belonging to $\Theta$.

Fortunately, it is not difficult to transform the variables $(\alpha_I,\theta^J)$ to ones which cannot define Lorentzian metrics and which naturally provide d.o.f. satisfying all Assumptions. Indeed, it is easy to realize that the only source of the two problems with the variables $(\alpha_I,\theta^J)$ is the fact that every triplet $(\alpha_I)$ corresponding to an element of $\Theta$ defines a function on $\Sigma$ valued in the ball $\mathbb{B}$. Thus to remove the problems it is enough to choose a diffeomorphism from $\mathbb{B}$ onto $\R^3$,  
\[
\mathbb{B}\ni(z_J)\mapsto \tau(z_J)=\Big(\tau^1(z_J),\tau^2(z_J),\tau^3(z_J)\Big)\in\R^3
\]
and define new variables as
\[
(\tau^I(\alpha_K),\theta^J)
\]
and 
\begin{equation}
\Theta\ni\theta\mapsto\kappa^I_y(\theta):=\tau^I(\alpha_J(y))\in\R
\label{tau-dof}
\end{equation}
as a new elementary d.o.f. instead of \eqref{al-dof}. 

However, there are many diffeomorphisms of this sort and the question is which one should we use? Or, is there a distinguished diffeomorphism? As we will show below a pair of such diffeomorphisms is distinguished by an ADM-like Hamiltonian framework of TEGR. 

\subsection{New variables and new d.o.f.}

Given $(\theta^A)\in\Theta$, consider the following equations imposed on smooth functions $\xi^A$ ($A=0,1,2,3$) on $\Sigma$ \cite{nester}:
\begin{align}
&\xi^A\theta_A=0, &&\xi^A\xi_A=-1.
\label{xi-df}
\end{align}
Solutions of these equations play an important role in deriving an ADM-like Hamiltonian framework of TEGR \cite{nester,oko-tegr} and YMTM \cite{os}---the configuration variable of Lagrangian formulations of TEGR and YMTM is a cotetrad field on a four-dimensional manifold; the cotetrad field is decomposed into ``time-like'' and ``space-like'' parts the latter one being $(\theta^A)\in\Theta$; then a solution of \eqref{xi-df} is used to express the ``time-like'' part as a function of the ADM lapse function, the ADM shift vector field and $(\theta^A)$. Moreover, a solution of \eqref{xi-df} appears in formulae describing constraints of both TEGR and YMTM, and the equations \eqref{xi-df} are used repeatedly while deriving constraint algebras of both theories \cite{oko-tegr,oko-tegr-der,os}.

Note that at every point $y\in\Sigma$ the values $(\xi^A(y))$ of a solution of \eqref{xi-df} form a time-like vector in $\mathbb{M}$ which means that the value $\xi^0(y)$ cannot be $0$. Taking into account the assumed smoothness of $\xi^A$ we can expect that there exist exactly two distinct solutions of \eqref{xi-df} which can be distinguished by the sign of $\xi^0$. As shown in \cite{os} by presenting explicite solutions of \eqref{xi-df} the expectation is correct.

Surprisingly, it turns out that there is a simple relation between the variables $(\alpha_I)$ and the space-like components $\xi^I$ of $\xi^A$ being a solution of \eqref{xi-df} and this relation provides us with two diffeomorphisms of the sort we need. Indeed, taking into account Equation \eqref{th0-ath} we see that
\begin{equation}
(\xi^A)\equiv(\xi^0,\xi^I)=\xi^0(1,\alpha^I)
\label{xi-al}
\end{equation}
satisfy the first equation \eqref{xi-df}. Setting this result to the second equation \eqref{xi-df} we obtain
\[
\xi^0=\pm\frac{1}{\sqrt{1-\alpha_J\alpha^J}}=\frac{\sgn(\xi^0)}{\sqrt{1-\alpha_J\alpha^J}},
\]
where $\sgn(\xi^0)=\pm 1$ is the sign of $\xi^0$. Thus
\[
(\xi^A)=\frac{\sgn(\xi^0)}{\sqrt{1-\alpha_J\alpha^J}}(1,\alpha^I)
\]     
and 
\begin{equation}
\xi^I=\sgn(\xi^0)\frac{\alpha^I}{\sqrt{1-\alpha_J\alpha^J}}.
\label{BR3-diff}
\end{equation}
Clearly, the r.h.s. of the equation above defines two diffeomorphisms from $\mathbb{B}$ onto $\R^3$ 
\begin{equation}
\mathbb{B}\ni(z_J)\mapsto\tau(z_J):=\sgn(\xi^0)\Big(\frac{z^I}{\sqrt{1-z_Lz^L}}\Big)\in\R^3
\label{tau}
\end{equation}
and both seem to be equally well suited for our goal. 
     
In this way we obtain new variables on the Hamiltonian configuration space $\Theta$: 

\begin{lm}
Given function $\iota$ defined on the space of all global coframes on $\Sigma$ and valued in the set $\{1,-1\}$, there exists a one-to-one correspondence between elements of $\Theta$ and all sextuplets $(\xi^I,\theta^J)$ consisting of
\begin{enumerate}
\item functions $\xi^I$, $I=1,2,3$, on $\Sigma$, 
\item one-forms  $\theta^J$, $J=1,2,3$, on $\Sigma$  constituting a global coframe on the manifold. 
\end{enumerate}  
The correspondence is given by 
\begin{equation}
(\xi^I,\theta^J)\mapsto \Big(\theta^0=\iota(\theta^L)\frac{\xi_I}{\sqrt{1+\xi_K\xi^K}}\theta^I,\theta^J\Big)\in\Theta.
\label{th-bij-eq}
\end{equation}
\label{th-bij}
\end{lm}

\begin{proof}
Given $\sgn(\xi^0)=\pm 1$, the map 
\begin{equation}
\R^3\ni (z^I)\mapsto \tau^{-1}(z^I)=\sgn(\xi^0)\Big(\frac{z_I}{\sqrt{1+z_Lz^L}}\Big)\in\mathbb{B}
\label{tau-1}
\end{equation}
is the inverse of the diffeomorphism \eqref{tau} and therefore the map 
\begin{equation}
(\xi^I)\mapsto \sgn(\xi^0)\Big(\frac{\xi_I}{\sqrt{1+\xi_L\xi^L}}\Big)
\label{xi->al}
\end{equation}
describes a one-to-one correspondence between all triplets $(\xi^I)$ of real functions on $\Sigma$ and all triplets $(\alpha_I)$ of real functions on the manifold such that $\alpha_I\alpha^I<1$. Consequently, given function $\iota$, the map 
\begin{equation}
(\xi^I,\theta^J)\mapsto \Big(\iota(\theta^L)\frac{\xi_I}{\sqrt{1+\xi_L\xi^L}},\theta^J\Big)
\label{xith->}
\end{equation}
is a bijection from the set of all pairs $(\xi^I,\theta^J)$ as described in Lemma \ref{th-bij} onto the set of all pairs $(\alpha_I,\theta^J)$ as described in Lemma \ref{th-bij-0}---the inverse map to \eqref{xith->} reads
\begin{equation}
(\alpha_I,\theta^J)\mapsto \Big(\iota(\theta^L)\frac{\alpha^I}{\sqrt{1-\alpha_L\alpha^L}},\theta^J\Big).
\label{al->xi}
\end{equation}

To finish the proof it is enough to note that the composition of the bijection \eqref{xith->} with the bijection \eqref{alth->} gives the map \eqref{th-bij-eq}.
     
\end{proof}

Let us emphasize that Lemma \ref{th-bij} describes a family of distinct variables $(\xi^I,\theta^J)$ which differ from each other by the choice of the function $\iota$. To understand the role of the function $\iota$ let us fix both the function and a pair $(\xi^I,\theta^J)$ such that $\xi^I\xi_I$ is not the zero function on $\Sigma$. Solving Equations \eqref{xi-df} given by the quadruplet $(\theta^A)\in\Theta$ corresponding to $(\xi^I,\theta^J)$ via \eqref{th-bij-eq} we obtain two solutions $\xi^A$ which differ from each other by the sign of $\xi^0$. Note now that the one form $\theta^0$ can be expressed in terms the solutions $\xi^A$---using \eqref{xi->al} being the inverse map to one defined by \eqref{BR3-diff} we obtain
\[
\theta^0=\alpha_I\theta^I=\sgn(\xi^0)\frac{\xi_I}{\sqrt{1+\xi_K\xi^K}}\theta^I.
\]
Comparing this with \eqref{th-bij-eq} we conclude that the variables $(\xi^I)$ constituting the fixed pair $(\xi^I,\theta^J)$ coincide with space-like components of this solution $\xi^A$ for which $\sgn(\xi^0)=\iota(\theta^L)$. Thus the function $\iota$ allows us to relate unambiguously the variables $(\xi^I)$ to components of one of the two solutions of \eqref{xi-df}.   
 
Consider new variables $(\xi^I,\theta^J)$ given by a function $\iota$. Now we can express the formula \eqref{tau-dof} defining new elementary d.o.f. in the following form
\begin{equation}
\Theta\ni\theta\mapsto \kappa^I_y(\theta)=\xi^I(y)\in\R,
\label{xi-dof}
\end{equation}
where $y\in\Sigma$. 
Let 
\[
{\cal K}:=\{\ \kappa^{I}_y, \kappa^J_e \ \},
\]  
where $I,J=1,2,3$, $y$ runs over $\Sigma$, $e$ over the set of all edges in the manifold and $\kappa^J_e$ is given by \eqref{th-dof}. We choose ${\cal K}$ to be a set of configurational elementary d.o.f. generated by the variables $(\xi^I,\theta^J)$.

Taking into account the properties of the set $\bar{\cal K}'$ described in the previous subsection and \eqref{xi-dof} we immediately conclude that the new set $\cal K$ satisfies Assumptions \ref{as-sep}, \ref{as-poly} and \ref{as-bgi}. Diffeomorphisms on $\Sigma$ act on elements of $\cal K$ as they do on ones of $\bar{\cal K}'$ hence we have
\begin{align*}
\varphi^*\kappa^I_y&=\kappa^I_{\varphi(y)}, &\varphi^*\kappa^J_e=\kappa^J_{\varphi(e)}.
\end{align*}
Obviously, $\cal K$ is preserved by the action.      

Corollary \ref{cor-Kug} and the relation between $(\xi^I)$ and $(\alpha_J)$ (see \eqref{xith->} and \eqref{al->xi}) allows us to formulate the following lemma:
\begin{lm}
Let $u=\{y_1,\ldots,y_M\}$ be a finite collection of points in $\Sigma$ and $\gamma=\{e_1,\ldots,e_N\}$ be a graph such that either $u$ or $\gamma$ is not an empty set ($N,M\geq0$ but $N+M>0$). Then for every $(z^I_{i},x^J_{j})\in\R^{3M}\times \R^{3N}$ there exists $\theta\in\Theta_+(\Theta_-)$ such that  
\begin{align*}
&\kappa^I_{y_{i}}(\theta)=z^I_{i},&&\kappa^J_{e_{j}}(\theta)=x^J_{j}
\end{align*}
for every $I,J=1,2,3$, $i=1,\ldots,M$ and $j=1,2,\ldots,N$.
\label{lm-Kug-xi}
\end{lm}
\noindent Thus for every finite set 
\begin{equation}
{K}_{u,\gamma}:=\{ \ {\kappa}^I_{y_1},\ldots,{\kappa}^I_{y_M},\kappa^J_{e_1},\ldots,\kappa^J_{e_N} \ | \ I,J=1,2,3 \ \}.
\label{K-ug1}
\end{equation}
the map $\tilde{K}_{u,\gamma}$ given by \eqref{k-inj} is bijective. In other words, 
\[
\Theta_{K_{u,\gamma}}\cong\R^{3M}\times\R^{3N}.
\]
We conclude that the set $\cal K$ with a directed\footnote{The relation $\geq$ is defined as described just below the formula \eqref{K-ug}.} set $(K_{u,\gamma},\geq)$ given by all finite subsets of $\Sigma$ and all graphs in the manifold  meets Assumption \ref{as-bij}. 

Let us finally make sure that the new variables $(\xi^I,\theta^J)$ cannot define Lorentzian metrics on $\Sigma$. By virtue of \eqref{q-th^I} and \eqref{xith->} the metric $q$ on $\Sigma$ given by the variables  can be expressed as
\begin{equation}
q=q_{IJ}\theta^I\ot\theta^J=\Big(\delta_{IJ}-\frac{\xi_I\xi_J}{1+\xi^K\xi_K}\Big)\theta^I\ot\theta^J.
\label{q-xi}
\end{equation}
The eigenvalues $1$ and $1-\alpha_I\alpha^I$ of the matrix $(q_{IJ})$ found in the proof of Lemma \ref{al-th} (see Equation \eqref{1-aa}) expressed in terms of $(\xi^I)$ read $1$ and $(1+\xi_I\xi^I)^{-1}$. Consequently, the matrix $(q_{IJ})$ is positive definite for every $(\xi^I)$. Thus even if a triplet $(\theta^J)$ is not a global coframe on $\Sigma$ the corresponding metric $q$ satisfy
\[
q(Y,Y)\geq 0
\]           
for every vector $Y$ tangent to $\Sigma$. This means that the new variables cannot describe any metric on the manifold which locally or globally is Lorentzian.

It is also worth to note that if $\xi^I=0$ then $(\theta^J)$ is an orthonormal coframe with respect to $q$---this fact can be easily deduced from \eqref{q-xi}. Thus we can regard $(\xi^I)$ as variables indicating how much the coframe $(\theta^J)$ deviates from being orthonormal with respect to $q$ (of course, the same can be said about $(\alpha_I)$).    

We conclude that for every function $\iota$ the set $\cal K$ of d.o.f. defined by corresponding new variables $(\xi^I,\theta^J)$ satisfy all Assumptions listed in Section \ref{intro}. Moreover, the variables cannot define Lorentzian metrics on $\Sigma$.

\section{Summary}

In this paper we showed that the natural variables $(\theta^A)$ on the Hamiltonian configuration space $\Theta$ of TEGR (and YMTM) can be used to build via the general method described in \cite{q-nl} the space $\bar{\D}$ of kinematic quantum states. The space $\bar{\D}$ is constructed in a background independent manner. It turned out that states constituting this space correspond not only to elements of $\Theta$, but also to quadruplets $(\theta^A)$ which define Lorentzian metrics on the manifold $\Sigma$ being a space-like slice of a spacetime. Since the task of isolating quantum states in $\bar{\D}$ which do not correspond to Lorentzian metrics seems to be very difficult we decided to look for other more suitable variables.

The results of our inquiry is the family $\{(\xi^I,\theta^J)\}$ of variables  parameterized by functions $\{\iota\}$ defined on the set of all global coframes on $\Sigma$ and valued in the set $\{-1,1\}$. Each element of the family satisfies all Assumptions presented in Section \ref{intro} and cannot define any Lorentzian metric on $\Sigma$. Therefore we expect  that at least some of the variables can be used to define in a background independent way a space of kinematic quantum states for TEGR free of the undesired property of the space $\bar{\D}$. We will show in \cite{q-stat} by an explicite construction that the expectation is correct. 

However, at this moment we are not completely ready for a construction of a space of quantum states from variables $(\xi^I,\theta^J)$ because of the following reason. Recall that we would like to apply the Dirac's approach to a canonical quantization of TEGR which  means that once a space of kinematic quantum states is constructed we will have to impose on the states ``quantum constraints'' as counterparts of constraints on the phase space of TEGR---this is the second step of the Dirac's quantization procedure. The problem is that it is not obvious whether every element of the family $\{(\xi^I,\theta^J)\}$ generates a space of quantum states suitable for defining ``quantum constraints'' on it. 

Although at this stage we are not able to solve this problem completely, we will address the issue in the accompanying paper \cite{ham-nv}---we will show there that indeed some variables $(\xi^I,\theta^J)$ are quite problematic. Namely, the constraints of TEGR (and YMTM) when expressed in terms of these variables depend on a special function defined on $\Theta$. It turns out that this function cannot be even approximated by functions on any $\Theta_{K_{u,\gamma}}$. This means that in the case of a space of kinematic quantum states built from such variables we will not be able to define ``quantum constraints'' by means of a family of restrictions such that each restriction is imposed on elements of a single space $\D_\lambda$. It is clear that if we are not able to define ``quantum constraints'' in such a way then this task becomes much more difficult. 

Fortunately, as it will be proven in \cite{ham-nv}, there exist exactly two closely related elements of the family $\{(\xi^I,\theta^J)\}$ for which the problem just described does not appear---the elements are closely related in this sense that functions $\{\iota\}$ distinguishing them differ from each other by a factor $-1$. Using one of these two elements we will construct in \cite{q-stat} a space $\D$ of kinematic quantum states for TEGR. The space $\D$ will be obviously free of the undesired property of $\bar{\D}$ and we hope that $\D$ will be also suitable for carrying out the second step of the Dirac's procedure.

\paragraph{Acknowledgments} This work was partially supported by the grant N N202 104838 of Polish Ministerstwo Nauki i Szkolnictwa Wy\.zszego.

\appendix

\section{Proof of Lemma \ref{theta-x3} \label{theta-x3-prf}}

\subsection{Preliminaries \label{app-prel}}

Let us recall a definition of a {\em simple edge}---it is a one-dimensional connected $C^\infty$ submanifold of $\Sigma$ with two-point boundary. On the other hand, a one-dimensional $C^\infty$ submanifold of $\Sigma$ with boundary is a subset $E$ of the manifold such that for every $x\in E$ there exists a neighborhood $U$ of $x$ open in $\Sigma$ and a $C^\infty$ coordinate chart $\chi$ on $U$ such that
\[
\chi(E\cap U)=\{\ (z^1,z^2,z^3)\in\R^3 \ | \ z^1=z^2=0, \ 0<z^3<1\ \}
\]      
or
\[
\chi(E\cap U)=\{\ (z^1,z^2,z^3)\in\R^3 \ | \ z^1=z^2=0, \ 0\leq z^3<1\ \}.
\]      
Consequently, given an {\em oriented} simple edge $e$ there exist numbers $a<0$ and $b>1$ and a smooth curve 
\[
]a,b[\ni\lambda\mapsto\tilde{e}(\lambda)\in\Sigma
\]   
such that $(i)$ $e=\tilde{e}([0,1])$, $(ii)$ the orientations of the curve and the edge coincide, $(iii)$ the vector $\dot{\tilde{e}}(\lambda)$ tangent to the curve at $\tilde{e}(\lambda)$ is non-zero for every $\lambda\in[0,1]$. Such a curve will be called {\em standard curve} for $e$. If $\tilde{e}$ is a standard curve for a simple edge $e$ and $\omega$ is a smooth one-form on $\Sigma$ then a map        
\[
]a,b[\ni\lambda \mapsto \omega(\dot{\tilde{e}}(\lambda))\in\R
\]
is smooth.

An edge is an {\em oriented} one-dimensional connected $C^0$ submanifold of $\Sigma$ given by a finite union of simple edges. Given an edge $e$ of two-point boundary, its orientation  allows to distinguish one of its endpoints as {\em a source} denoted by $\src$ and the other as {\em a target} denoted by $\tar$; if an edge is a loop then we choose one of its points and treat it as both the source and the target of the edge. We will call the set $e\setminus\{\src,\tar\}$ {\em interior} of the edge $e$ and will denote it by $\intr e$. Note that, given an oriented simple edge $e$ and its standard curve $\tilde{e}$, $\src=\tilde{e}(0)$, $\tar=\tilde{e}(1)$ and $\intr e=\tilde{e}(]0,1[)$. 

An edge $e$ is a composition of edges $e_1$ and $e_2$, $e=e_2\circ e_1$, if $(i)$ $e$ as an oriented manifold is a union of $e_1$ and $e_2$, $(ii)$ $\tar_1=\src_2$, $(iii)$ $e_1\cap e_2$ consists merely of some (or all) endpoints of $e_1$ and $e_2$. Every edge turns out to be a composition of oriented simple edges. 

Given a set $W\subset\Sigma$, its characteristic function is equal $1$ on $W$ and $0$ outside $W$. We will call a function $\phi:\Sigma\mapsto\R$ an {\em almost characteristic function of W} if it is {\em smooth} and is zero outside $W$, positive on $W$  and if there exists a non-empty open subset $W'\subset W$ such that $\phi$ is equal $1$ on $W'$. We will also denote by $\bld{1}$ a constant function on $\Sigma$ of values equal $1$.

To simplify the notation, {given quadruplet $(\theta^A)\in{\Theta}$},  we introduce the following symbol
\[
\theta^A(e)\equiv\int_e\theta^A=\kappa^A_e(\theta).
\]
For every composition $e_2\circ e_1$ 
\begin{equation}
\theta^A(e_2\circ e_1)=\theta^A(e_1)+\theta^A(e_2).
\label{ee-omega}
\end{equation}

\subsection{The proof}

Proving Lemma \ref{theta-x3} amounts to proving the following one:

\begin{lm}
Let $\gamma$ be a graph in $\Sigma$ consisting of oriented simple edges $\{e_1,\ldots,e_N\}$ and let $(\theta^I)$ be a smooth global coframe on $\Sigma$. Fix an edge $e_{i}$ of $\gamma$ and a non-zero vector $(x^I)\in\R^3$. Then the coframe $(\theta^I)$ can be deformed to a smooth global coframe $(\bar{\theta}^I)$ such that $(\bar{\theta}^I)$ defines the same orientation of $\Sigma$ as $(\theta^I)$ does and    
\[
\bar{\theta}^I(e_{j})=
\begin{cases}
x^I & \text{if $j=i$}\\
\theta^I(e_{j}) & \text{otherwise}
\end{cases}.
\]    
\label{theta-x3-nz}     
\end{lm}

\begin{proof}
Let $U$ be an open subset of $\Sigma$ such that 
\[
\gamma\cap U=\intr e_{i}.
\]

The main part of the proof will be divided into four steps:
\begin{enumerate}
\item first we will divide the edge $e_{i}$ into three edges $f_0,f_1$ and $f_2$ such that
$e_{i}=f_2\circ f_0\circ f_1$ and then we will fix an open set $U_0\subset U$ such that
\begin{equation}
U_0\cap f_j=
\begin{cases}
\intr f_j & \text{if $j=0$},\\
\varnothing & \text{otherwise}
\end{cases}.
\label{ua}
\end{equation}
\item then we will modify the coframe $(\theta^I)$ on the set $U$ in such a way that the resulting coframe $(\theta^{\prime I})$ will satisfy
\begin{equation}
\theta^{\prime I}(f_0)=x^I.
\label{f0-xI}
\end{equation}
\item next we will deform the coframe $(\theta^{\prime I})$ on $U\setminus f_0$ obtaining thereby a coframe $(\theta^{\prime\prime I})$ for which the vector
\[
\Big(\theta^{\prime\prime I}(f_1)+\theta^{\prime\prime I}(f_2)\Big)\in\R^3
\]  
will meet a special condition.
\item finally, we will modify $(\theta^{\prime\prime I})$ on $U_0$ in such a way that the resulting coframe $(\bar{\theta}^I)$ will satisfy
\[
\bar{\theta}^{I}(f_0)=x^I-\theta^{\prime\prime I}(f_1)-\theta^{\prime\prime I}(f_2).
\] 
\end{enumerate} 
This will finish the main part of the proof since for $(\bar{\theta}^I)$ constructed in this way
\[
\bar{\theta}^I(e)=\bar{\theta}^I(f_1)+\bar{\theta}^I(f_0)+\bar{\theta}^I(f_2)=x^I.
\]

\begin{figure}
\psfrag{U}{$U'$}
\psfrag{e}{$e$}
\psfrag{ch}{$\chi_\lambda$}
\psfrag{y0}{$y_0$}
\psfrag{l}{$\lambda\!-\!\lambda_0$}
\psfrag{ch0}{$\chi_{\lambda_0}$}
\psfrag{s}{$s$}
\psfrag{z}{$(z^1,z^2)=(\lambda\!-\!\lambda_0,s)$}
\begin{center}
\includegraphics[scale=0.7]{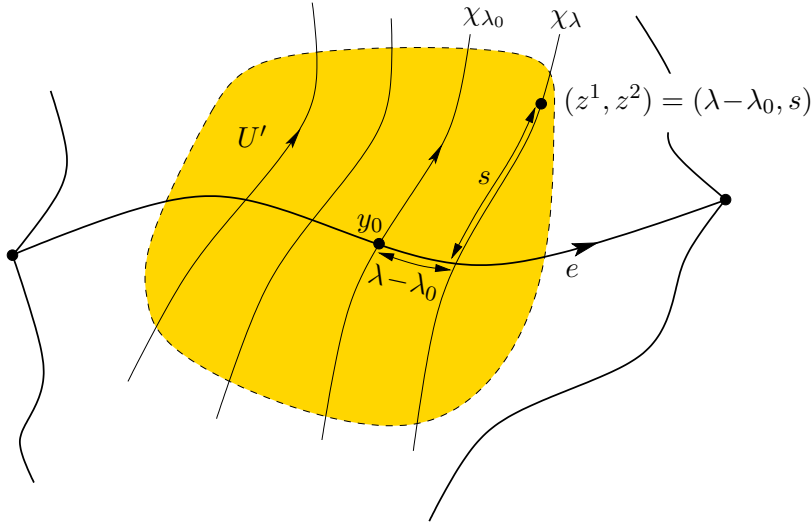}
\end{center}
\caption{Construction of the coordinate frame $(z^1,z^2)$.}
\label{fig-zz}
\end{figure}

\paragraph{Step 1} Denote for simplicity $e_{i}\equiv e$ and fix a standard curve $\tilde{e}$ for the edge.

Let $(Y_I)$ be a frame on $\Sigma$ dual to $(\theta^I)$. Then $Y:=x^IY_I$, where $(x^I)$ is the fixed vector in $\R^3$, is a {\em non-zero} vector field on the manifold. Assume that there exists a point $y_0\in\intr e$ such that the value of $Y$ at $y_0$ {\em is not} tangent to $e$---this assumption will allow us to construct a special coordinate frame on a neighborhood $U'\subset U$ of $y_0$.

To this end consider a bunch of integral curves of the vector field $Y$ which intersect the edge $e$ at points belonging to $U'$---see Figure \ref{fig-zz}. This bunch can be parameterized by the parameter $\lambda$ of the curve $\tilde{e}$: 
\[
\R\ni{s}\mapsto\chi_\lambda({s})\in\Sigma
\]
is an integral curve of $Y$ which intersects\footnote{Of course, we have to choose the neighborhood $U'$ ``small'' enough to ensure that every integral curve in the bunch intersects the set $e\cap U'$ exactly once.} the edge $e$ at the point $\tilde{e}(\lambda)$. Moreover, we can adjust the parameter ${s}$ along each integral curve in the bunch in such a way that $\chi_\lambda({s}=0)$ coincides with the intersection point i.e.   $\chi_\lambda({s}=0)= \tilde{e}(\lambda)$.

Now, if a point $y$ lies on the curve $\chi_\lambda$  i.e. if $y=\chi_\lambda({s})$ then we can associate with it two numbers: 
\[
z^1=\lambda-\lambda_0 \ \ \ \text{and} \ \ \ z^2={s},
\]
where $\tilde{e}(\lambda_0)=y_0$. Thus we obtained a coordinate frame $(z^1,z^2)$ on the bunch. If $U'$ is sufficiently ``small'' then one can find a function $z^3$ on $U'$ such that its values are zero on the bunch and $(z^1,z^2,z^3)$ are coordinates on $U'$. There exists a positive number 
\begin{equation}
\zeta<{\rm min}\{\lambda_0,1-\lambda_0\}\leq\frac{1}{2}
\label{zeta<}
\end{equation}
such that the values of each coordinate {in $(z^1,z^2,z^3)$} ranges at least between $-\zeta$ and $\zeta$.

\begin{figure}
\psfrag{ch0}{$\chi_{\lambda_0}$}
\psfrag{U'}{$U_1$}
\psfrag{U}{$U_0$}
\psfrag{f0}{$f_0$}
\psfrag{f1}{$f_1$}
\psfrag{f2}{$f_2$}
\begin{center}
\includegraphics[scale=0.7]{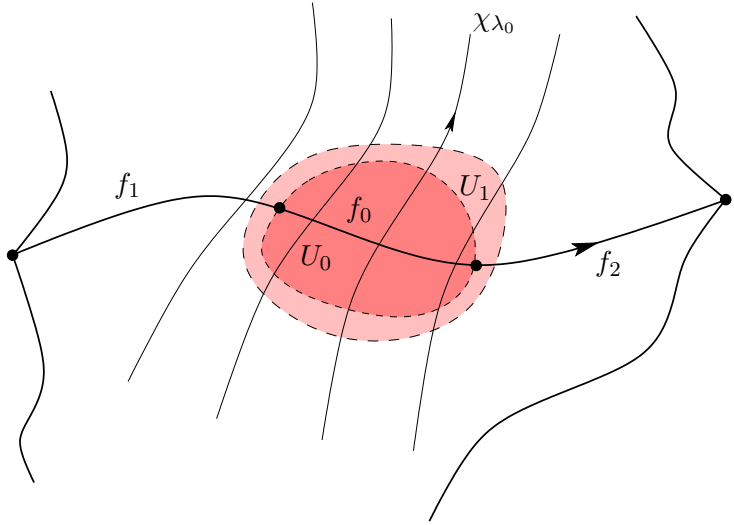}
\end{center}
\caption{The edges $f_0$, $f_1$ and $f_2$ and the sets $U_0$ and $U_1$.}
\label{fff-fig}
\end{figure}

Let us fix a number $0<r<\zeta$ and define the edges   
\[
f_0:=\tilde{e}([\lambda_0-r,\lambda_0+r]), \ \ f_1:=\tilde{e}([0,\lambda_0-r]), \ \ f_2:=\tilde{e}([\lambda_0+r,1]).
\]  
and sets
\begin{align*}
U_0&:=\{\ (z^1,z^2,z^3)\in U' \ |\ (z^1)^2+(z^2)^2+(z^3)^2<r^2\ \},\\
U_1&:=\{\ (z^1,z^2,z^3)\in U' \ |\ (z^1)^2+(z^2)^2+(z^3)^2<(r')^2\ \},
\end{align*}
where $r<r'<\zeta$---see Figure \ref{fff-fig}.

\paragraph{Step 2} Let $Z_0$ be a vector field on $U'$ defined as
\[
Z_0=z^2\partial_{z^1}-z^1\partial_{z^2}.
\]   
\begin{figure}
\psfrag{U}{$U_0$}
\psfrag{tau}{$\tau_{\pi/2}(e)$}
\psfrag{ch0}{$\chi_{\lambda_0}$}
\begin{center}
\includegraphics[scale=0.7]{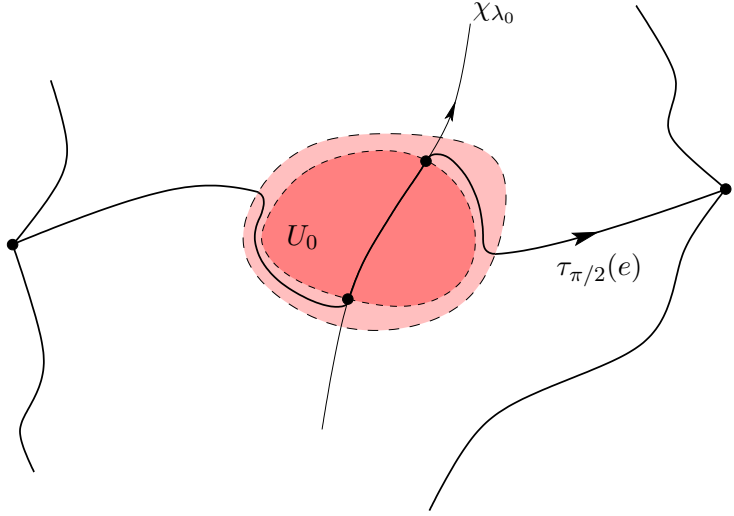}
\end{center}
\caption{The action of $\tau_{\pi/2}$ on the edge $e$.}
\label{fig-tau}
\end{figure}
\noindent If $\phi$ is an almost characteristic function on $U_1$ such that it is equal $1$ on $U_0$ then the local vector field $\phi Z_0$ can be naturally extended to a smooth vector field $Z$ on $\Sigma$ of a compact support which is the closure of $U_1$. Let $\{\tau_t\}_{t\in\R}$ be a one-parameter family of diffeomorphisms on $\Sigma$ generated by $Z$. Clearly, in the coordinate frame $(z^1,z^2,z^3)$ the restriction of diffeomorphism  $\tau_t$ to $U_0$ is a rotation around the $z^3$-axis through an angle $t$. Hence $\tau_{\pi/2}$ maps the $z^1$ axis onto $z^2$ axis, that is, the edge $f_0$ into the image of the integral curve $\chi_0$---see Figure \ref{fig-tau}. Moreover, since $z^1$ is the parameter along the curve $\tilde{e}$ and $z^2$ is the parameter along $\chi_{\lambda_0}$       
\[
\tau_{\pi/2}(\tilde{e}(\lambda))=\chi_{\lambda_0}(\lambda-\lambda_0)
\]                   
provided $\lambda\in [\lambda_0-r,\lambda_0+r]$. Therefore for every $\lambda\in [\lambda_0-r,\lambda_0+r]$ the tangent vector $\dot{\tilde{e}}(\lambda)$  satisfies 
\[
\tau_{\pi/2*}(\dot{\tilde{e}}(\lambda))=Y(\chi_{\lambda_0}(\lambda-\lambda_0)).
\]
Consequently,
\begin{equation}
(\tau_{\pi/2}^*\theta^I)(\dot{\tilde{e}}(\lambda))=\theta^I(Y)=x^I
\label{tau-th-vec}
\end{equation}
and
\[
\int_{f_0}(\tau^*_{\pi/2}\theta^I)=\int_{\lambda_0-r}^{\lambda_0+r}x^I\,d\lambda=2rx^I.
\] 

Let $\phi$ be the almost characteristic function on $U_1$ equal $1$ on $U_0$. Then the following one-forms
\begin{equation}
\theta^{\prime I}:=(\bld{1}+(\frac{1}{2r}-1)\phi)\tau^*_{\pi/2}\theta^I 
\label{th'}
\end{equation}
form a global coframe on $\Sigma$---note that by virtue of the inequalities $r<\zeta$ and \eqref{zeta<} the function $(\bld{1}+(\frac{1}{2r}-1)\phi)$ is positive. This coframe coincides with $(\theta^I)$ outside the set $U_1\subset U$ and satisfies \eqref{f0-xI}  

\begin{figure}
\psfrag{U}{$U_0$}
\psfrag{g0}{$g_0$}
\psfrag{g1}{$g_1$}
\psfrag{g2}{$g_2$}
\psfrag{e1}{$\epsilon_1$}
\psfrag{e2}{$\epsilon_2$}
\psfrag{W}{$\bigcup_\beta W_\beta$}
\psfrag{W'}{$\bigcup_{\beta} W'_{\beta}$}
\begin{center}
\includegraphics[scale=0.7]{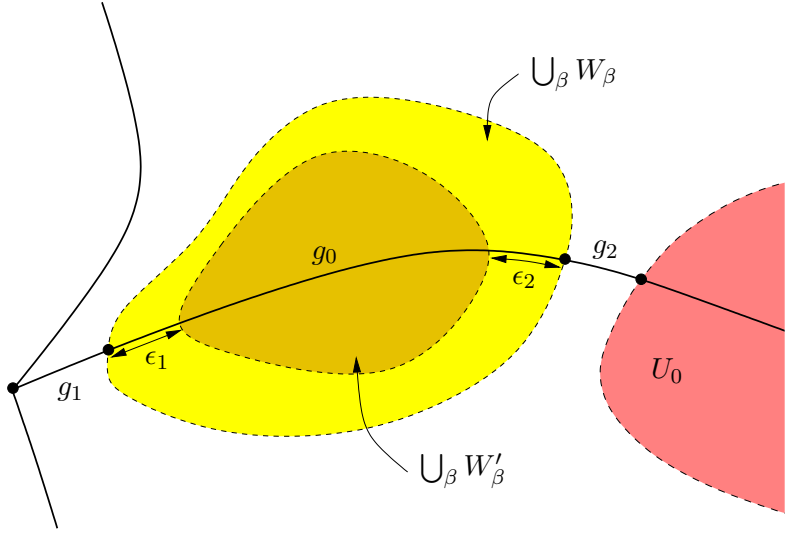}
\end{center}
\caption{Construction of $(\theta^{\prime\prime I})$ on a neighborhood of the edge $f_1$.}
\label{f1-fig}
\end{figure}

\paragraph{Step 3} We assumed that $(x^I)$ is a non-zero vector in $\R^3$. Without loss of generality we can assume that $x^1\neq 0$. Our goal now is to deform the coframe $(\theta^{\prime I})$ on the set $U\setminus f_0$ in such a way that the resulting coframe $(\theta^{\prime\prime I})$ satisfies     
\begin{equation}
|\theta^{\prime\prime 1}(f_1)+\theta^{\prime\prime 1}(f_2)|<\frac{|x^1|}{3}.
\label{tt<x3}
\end{equation}

To this end we divide the edge $f_1$ into edges $g_0,g_1,g_2$ such that $f_1=g_2\circ g_0\circ g_1$---see Figure \ref{f1-fig}---and choose the edges $g_1$ and $g_2$ to be short enough to satisfy
\[
|\theta^{\prime 1}(g_1)|<\frac{|x^1|}{18}\ \ \ \text{and} \ \ \  |\theta^{\prime 1}(g_2)|<\frac{|x^1|}{18}.
\]

To carry out the desired deformation of the coframe $(\theta^{\prime I})$ we proceed as follows: by virtue of the compactness of $g_0$ we can cover $\intr g_0$ by a finite {number} of open subsets $\{W_{\beta}\}$ such that each $W_\beta$ admits existence of an almost characteristic function $\phi_\beta$ on it\footnote{To satisfy this requirement $W_\beta$ may be defined as an open  coordinate ball of non-zero radius.} and
\[
\gamma\cap \Big(\bigcup_\beta W_\beta\Big)=\intr g_0. 
\]        

Let $a_1,a_2\in]0,\lambda_0-r[$ be numbers such that 
\[
g_0=\tilde{e}([a_1,a_2]).
\]  
For each almost characteristic function $\phi_\beta$ on $W_\beta$ we define $W'_\beta\subset W_\beta$ as a set on which $\phi_\beta$ is equal $1$. Since the union $\bigcup_\beta W_\beta$ covers $\intr g_0$  the sets $\{W'_\beta\}$ can be chosen in such a way that the union $\bigcup_\beta W'_\beta$ covers $\tilde{e}([a_1+\epsilon_1,a_2-\epsilon_2])$ for some $\epsilon_1,\epsilon_2>0$.

If a number $\nu$ satisfies $0<\nu<1$ then the function
\begin{equation}
\prod_\beta(\bld{1}-\nu\phi_\beta)
\label{f-prod}
\end{equation}
is positive on $\Sigma$, equal one outside the union $\bigcup_\beta W_\beta$ and it is not greater than $(1-\nu)$ on the union $\bigcup_\beta W'_\beta$ covering the edge $\tilde{e}([a_1+\epsilon_1,a_2-\epsilon_2])$. Therefore
\begin{multline}
\Big|\int_{g_0} \prod_\beta(\bld{1}-\nu\phi_\beta)\theta^{\prime 1}\Big|=\Big|\int_{a_1}^{a_2}\prod_\beta(\bld{1}-\nu\phi_\beta)\theta^{\prime 1}(\dot{\tilde{e}})\,d\lambda\Big|\leq\Big|\int_{a_1}^{a_1+\epsilon_1}\prod_\beta(\bld{1}-\nu\phi_\beta)\theta^{\prime 1}(\dot{\tilde{e}})\,d\lambda\Big|+\\+\Big|\int_{a_1+\epsilon_1}^{a_2-\epsilon_2}\prod_\beta(\bld{1}-\nu\phi_\beta)\theta^{\prime 1}(\dot{\tilde{e}})\,d\lambda\Big|+\Big|\int_{a_2-\epsilon_2}^{a_2}\prod_\beta(\bld{1}-\nu\phi_\beta)\theta^{\prime 1}(\dot{\tilde{e}})\,d\lambda\Big|
\label{int-g-0}
\end{multline}
and
\begin{multline}
\Big|\int_{a_1+\epsilon_1}^{a_2-\epsilon_2}\prod_\beta(\bld{1}-\nu\phi_\beta)\theta^{\prime 1}(\dot{\tilde{e}})\,d\lambda\Big|\leq \int_{a_1+\epsilon_1}^{a_2-\epsilon_2}\Big|\prod_\beta(\bld{1}-\nu\phi_\beta)\theta^{\prime 1}(\dot{\tilde{e}})\Big| d\lambda\leq \\ \leq(1-\nu)(a_2-\epsilon_2-a_1-\epsilon_1)\underset{\lambda\in [a_1+\epsilon_1,a_2-\epsilon_2]}{\rm supp}|\theta^{\prime 1}(\dot{\tilde{e}}(\lambda))|<\\<(1-\nu)(a_2-a_1)\underset{\lambda\in [a_1,a_2]}{\rm supp}|\theta^{\prime 1}(\dot{\tilde{e}}(\lambda))|
\label{rest}
\end{multline}

Note now that modifying appropriately the function \eqref{f-prod} we can  make the value of the l.h.s. of \eqref{int-g-0} as small as we want. This can be achieved by $(i)$ choosing $\nu$ as close to $1$ as we want and $(ii)$ choosing the functions $\{\phi_\beta\}$ in such a way that the sets $\{W'_\beta\}$ determine the values of $\epsilon_1,\epsilon_2$ as close to zero as we want. An important observation is that the restriction \eqref{rest} is independent of $\epsilon_1,\epsilon_2$. Consequently we can first choose  $\nu$ to restrict appropriately the value of the second term at the r.h.s. of \eqref{int-g-0} and then we can choose values of $\epsilon_1,\epsilon_2$ to restrict the values of the first and the third terms {\em without spoiling} the restriction imposed on the second one.         

An analogous construction done for the edge $f_2$ provides us with a function 
\[
\prod_{\beta'}(\bld{1}-\nu'\phi'_{\beta'}).
\]

Let
\[
\theta^{\prime \prime I}:=
\begin{cases}
\prod_\beta(\bld{1}-\nu\phi_\beta)\prod_{\beta'}(\bld{1}-\nu'\phi'_{\beta'})\,\theta^{\prime 1}, & \text{for $I=1$}\\
\theta^{\prime I} & \text{otherwise}
\end{cases}.
\]
For the number $\nu$ and the functions $\{\phi_\beta\}$ appropriately chosen 
\[
|\theta^{\prime \prime 1}(g_0)|=\Big|\int_{g_0}\prod_\beta(\bld{1}-\nu\phi_\beta)\theta^{\prime 1}\Big|<\frac{|x^1|}{18}.
\]         
Thus
\[
|\theta^{\prime \prime 1}(f_1)|<|\theta^{\prime \prime 1}(g_1)|+|\theta^{\prime \prime 1}(g_0)|+|\theta^{\prime \prime 1}(g_2)|<\frac{|x^1|}{6}.
\]
Similarly, we can choose the number $\nu'$ and the functions $\{\phi'_{\beta'}\}$ in such a way that
\[
|\theta^{\prime \prime 1}(f_2)|<\frac{|x^1|}{6}.
\]
In this way we obtained a coframe $(\theta^{\prime \prime I})$ which satisfies \eqref{tt<x3}.

\paragraph{Step 4} Recall that $U_0$ is an open set such that $\gamma\cap U_0=\intr f_0$. We keep assuming that $x^1\neq 0$. Denote 
\[
t^I\equiv \theta^{\prime\prime I}(f_1)+\theta^{\prime\prime I}(f_2)
\]    
and consider the following one-forms
\begin{equation}
\bar{\theta}^I=\theta^{\prime\prime I}-3\frac{t^I}{x^1}\,\phi \,\theta^{\prime\prime 1},
\label{b-th}
\end{equation}
where $\phi$ is an almost characteristic function on $U_0$ such that
\[
\int_{\lambda_0-r}^{\lambda_0+r} \phi(\tilde{e}(\lambda))\,d\lambda=\frac{2r}{3}.
\]  

Note that $(\bar{\theta}^I)$ is a global coframe---indeed, the determinant of a matrix defining the transformation \eqref{b-th} between $(\theta^{\prime\prime I})$ and $(\bar{\theta}^I)$ at a point $x\in\Sigma$ is equal to
\begin{equation}
1-3\frac{t^I}{x^1}\phi(x).
\label{tr-mtx-det}
\end{equation}
By virtue of \eqref{tt<x3}
\[
1>3\frac{|t^I|}{|x^1|}\geq 3\frac{|t^I|}{|x^1|}\phi(x)\geq 3\frac{t^I}{x^1}\phi(x).
\]
Consequently, the determinant  \eqref{tr-mtx-det} is positive and the transformation \eqref{b-th} is invertible at every point $x\in\Sigma$.

On the other hand, by virtue of \eqref{th'} and \eqref{tau-th-vec} for $\lambda\in[\lambda_0-r,\lambda_0+r]$ 
\[
{\theta}^{\prime\prime 1}(\dot{\tilde{e}}(\lambda))=\frac{x^1}{2r}
\]   
and
\begin{multline*}
\bar{\theta}^I(f_0)=\int_{f_0}\bar{\theta}^I=\int_{f_0}\theta^{\prime\prime I}-3\frac{t^I}{x^1}\int_{\lambda_0-r}^{\lambda_0+r} \phi\, \theta^{\prime\prime 1}(\dot{\tilde{e}})\,d\lambda=x^I-3\frac{t^I}{2r}\int_{\lambda_0-r}^{\lambda_0+r} \phi(\tilde{e}(\lambda))\,d\lambda=\\=x^I-t^I=x^I-\bar{\theta}^I(f_1)-\bar{\theta}^I(f_2).
\end{multline*}

This finishes the main part of the proof.

\paragraph{Final remarks} Recall that while carrying out Step 1 we were assuming that there exists a point $y_0\in\intr e$ such that the value of the vector field $Y=x^IY_I$ at this point is not tangent to $e$. If there is no such point then Step 1 should be preceded by a modification of the original coframe $(\theta^I)$ on the set $U$ which may consist in a pull-back of the coframe by means of a diffeomorphism similar to that applied in Step 2. Choosing appropriately the diffeomorphism one can obtain a coframe satisfying the assumption.      

Note also that each transformation of the coframes used in the proof preserves both smoothness of the coframes and the orientation of $\Sigma$ defined by the original coframe $(\theta^I)$.  
\end{proof}

Formulating Lemma \ref{theta-x3-nz} we assumed that the vector $(x^I)$ is non-zero. Let us divide the edge $e_{i}$ considered in the lemma into two edges $f_1$ and $f_2$ such that $e_{i}=f_2\circ f_1$. By virtue of the lemma the coframe $(\theta^I)$ can be deformed to a coframe $(\bar{\theta}^I)$ such that
\begin{align*}
\bar{\theta}^I(f_1)&=x^{\prime I},\\
\bar{\theta}^I(f_2)&=-x^{\prime I},\\
\bar{\theta}^I(e_{j})&={\theta}^I(e_{j}), \ \ \ j\neq i\\
\end{align*}
for some non-zero vector $(x^{\prime I})$. Thus 
\[
\bar{\theta}^I(e_{i})=-x^{\prime I}+x^{\prime I}=0.
\]

\begin{cor}
Lemma \ref{theta-x3-nz} holds also in the case of $(x^I)=0$. 
\end{cor} 

An immediate consequence of Lemma \ref{theta-x3-nz} and  the corollary above is Lemma \ref{theta-x3} restricted to graphs built from oriented simple edges only. But because every edge is a composition of simple edges and because of \eqref{ee-omega} Lemma \ref{theta-x3} it true without any restrictions.

\section{Proof of Lemma \ref{theta-x} \label{theta-x-prf}}

We will prove slightly stronger versions of Lemma \eqref{theta-x}---the versions are obtained by replacing in the lemma the condition $\theta\in\Theta$ by, respectively, $\theta\in\Theta_+$ and $\theta\in\Theta_-$.

While proving the lemma we will use the notation and some notions introduced in Section \ref{app-prel}.

In fact, it is enough to prove the lemma for every graphs built from {\em oriented simple} edges. Consider then a graph $\gamma$ being a collection $\{e_1,\ldots,e_N\}$ of such edges. Let us divide each edge $e_{{j}}$ of the graph into three edges $f_{{j}1},f_{{j}2}$ and $f_{{j}3}$ such that $e_{{j}}=f_{{j}3}\circ f_{{j}2}\circ f_{{j}1}$. Given $(x^A_{{j}})\in\mathbb{M}^N$, by virtue of Lemma \ref{theta-x3} there exists a global coframe $(\theta^I)$ on $\Sigma$ compatible (incompatible) with the orientation of the manifold such that for every ${j}=1,2,\ldots,N$ 
\begin{equation}
\begin{aligned}
\theta^1(f_{{j}1})&=-x^0_{{j}}, & \theta^1(f_{{j}2})&=x^0_{{j}}, & \theta^1(f_{{j}3})&=x^1_{{j}},\\
\theta^2(f_{{j}1})&=0, & \theta^2(f_{{j}2})&=0, & \theta^2(f_{{j}3})&=x^2_{{j}},\\
\theta^3(f_{{j}1})&=0, & \theta^3(f_{{j}2})&=0, & \theta^3(f_{{j}3})&=x^3_{{j}}. 
\end{aligned}
\label{th-I-an}
\end{equation}

To prove the lemma it is enough to find a smooth function $\alpha_1$ on $\Sigma$ such that at every point $x\in\Sigma$  
\begin{equation}
(\alpha_1(x))^2<1
\label{al<1}
\end{equation}
and 
\begin{align}
\int_{f_{{j}1}}\alpha_1\theta^1&=\frac{x^0_{{j}}}{2}, & \int_{f_{{j}2}}\alpha_1\theta^1&=\frac{x^0_{{j}}}{2},&\int_{f_{{j}3}}\alpha_1\theta^1&=0.
\label{th-0-an}
\end{align}
Indeed, if $\alpha_1$ satisfies both conditions above then choosing additionally functions $\alpha_2:=\alpha_3:=0$ and defining
\[
\theta^0:=\alpha_I\theta^I
\] 
we obtain a quadruplet $(\theta^0,\theta^I)$ of one-forms.  By virtue of Lemma \ref{al-th} the quadruplet is an element of $\Theta$. Moreover, if $(\theta^I)$ is compatible (incompatible) with the orientation of $\Sigma$ then the quadruplet belongs to  $\Theta_+$ ($\Theta_-$). Due to \eqref{th-I-an} and \eqref{th-0-an} 
\begin{align*}
\theta^0(e_{{j}})&=\int_{f_{{j}1}}\alpha_1\theta^1+\int_{f_{{j}2}}\alpha_1\theta^1+\int_{f_{{j}3}}\alpha_1\theta^1=\frac{x^0_{{j}}}{2}+\frac{x^0_{{j}}}{2}+0=x^0_{{j}},\\
\theta^{1}(e_{{j}})&=\theta^1(f_{{j}1})+\theta^1(f_{{j}2})+\theta^1(f_{{j}3})=-x^0_{{j}}+x^0_{{j}}+x^1_{{j}}=x^1_{{j}},\\
\theta^{2}(e_{{j}})&=\theta^2(f_{{j}1})+\theta^2(f_{{j}2})+\theta^2(f_{{j}3})=0+0+x^2_{{j}}=x^2_{{j}},\\
\theta^{3}(e_{{j}})&=\theta^3(f_{{j}1})+\theta^3(f_{{j}2})+\theta^3(f_{{j}3})=0+0+x^3_{{j}}=x^3_{{j}}.
\end{align*}

\begin{figure}
\psfrag{a+}{$\alpha_1>0$}
\psfrag{a-}{$\alpha_1<0$}
\psfrag{f1}{$f_{{j}1}$}
\psfrag{f2}{$f_{{j}2}$}
\psfrag{f3}{$f_{{j}3}$}
\psfrag{g1}{$g_{{j}1}$}
\psfrag{g2}{$g_{{j}2}$}
\psfrag{W1}{$\bigcup_\beta W^{{{j}1}}_\beta$}
\psfrag{W2}{$\bigcup_{\beta'} W^{{{j}2}}_{\beta'}$}
\begin{center}
\includegraphics[scale=0.7]{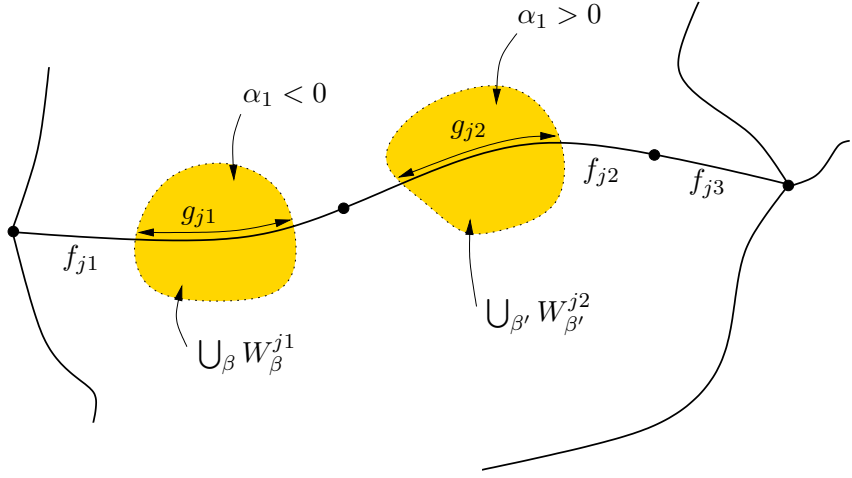}
\end{center}
\caption{Construction of $\alpha_1$ on a neighborhood of the edge $e_{{j}}$.}
\label{alpha-fig}
\end{figure}

Let us then start a construction of the desired function $\alpha_1$. Since now till the end of the proof we will exclude from our considerations the edges  $\{f_{{j}3}\}$ and focus solely on edges $\{f_{{j}a}\}$ with $a=1,2$. Moreover, since now till Equation \eqref{varphi=0} we will restrict ourselves to those edges $\{f_{{j}a}\}$ for which the corresponding $x^0_{{j}}\neq 0$. 

Let us fix a standard curve $\tilde{f}_{{j}a}$ for every edge $f_{{j}a}$ under consideration. For each $f_{{j}a}$ there exist numbers $\lambda_{-},\lambda_{+}\in[0,1]$ ($\lambda_-<\lambda_+$)  such that 
\begin{align*}
\int_0^{\lambda_-} \theta^1(\dot{\tilde{f}}_{{j}a})\,d\lambda&=(-1)^a\frac{1}{5}x^0_{{j}},\\
\int_0^{\lambda_+} \theta^1(\dot{\tilde{f}}_{{j}a})\,d\lambda&=(-1)^a\frac{4}{5}x^0_{{j}}, 
\end{align*}
where as before $\dot{\tilde{f}}_{{j}a}(\lambda)$ denotes  a vector tangent to the corresponding standard curve at the point $\tilde{f}_{{j}a}(\lambda)$. Let    
\[
g_{{j}a}:=\tilde{f}_{{j}a}([\lambda_-,\lambda_+]).
\]
Consequently,
\begin{equation}
\int_{g_{{j}a}}\theta^1=\int_{\lambda_-}^{\lambda_+} \theta^1(\dot{\tilde{f}}_{{j}a})\,d\lambda=(-1)^a\frac{3}{5}x^0_{{j}}.
\label{3/5}
\end{equation}

By virtue of the compactness of $g_{{j}a}$ we can cover $\intr g_{{j}a}$ by a finite {number} of open subsets $\{W^{{j}a}_\beta\}$ such that each $W^{{j}a}_\beta$ admits existence of an almost characteristic function $\phi^{{j}a}_\beta$ on it\footnote{To satisfy this requirement $W^{{j}a}_\beta$ may be defined as an open  coordinate ball of non-zero radius.} and
\begin{gather}
\gamma\cap {\Big(\bigcup_\beta W^{{j}a}_\beta\Big)}=\intr g_{{j}a},\label{g-W}\\
\overline{\Big(\bigcup_\beta W^{{j}a}_\beta\Big)}\cap\overline{\Big(\bigcup_{\beta'} W^{j'a'}_{\beta'}\Big)}=\varnothing\label{supports}
\end{gather}
if only ${j}\neq j'$ or $a\neq a'$---see Figure \ref{alpha-fig}. 

There exists a positive number $\zeta^{{j}a}$ such that the function   
\[
\phi^{{j}a}:=\zeta^{{j}a}\sum_\beta \phi^{{j}a}_\beta
\]
is {positive but lower than $1$ on $\bigcup_\beta W^{{j}a}_\beta$, is equal zero outside this set}  and  
\begin{equation}
\Big|\int_{g_{{j}a}}\phi^{{j}a}\theta^1\Big|<\frac{2}{5}|x^0_{{j}}|.
\label{2/5}
\end{equation}
Consider now a family of function
\[
\varphi^{{j}a}_{\sigma}:=\bld{1}-(\bld{1}-\phi^{{j}a})^\sigma,
\]
where the number $\sigma\geq 1$. It follows from the properties of $\phi^{{j}a}$ that for every allowed $\sigma$
\begin{equation}
0\leq \varphi^{{j}a}_{\sigma}< 1.
\label{varphi<1}
\end{equation}
Moreover, the function $\varphi^{{j}a}_{\sigma}$ is a smooth\footnote{Note that for $n>\sigma$ the $n$-th derivative of $\varphi^{{j}a}_{\sigma}$ contains a factor $(\bld{1}-\phi^{{j}a})^{\sigma-n}$ which could be a source of non-differentiability of $\varphi^{{j}a}_{\sigma}$ if a value of a function $(\bld{1}-\phi^{{j}a})$ was zero. This is, however, not the case---the function is positive everywhere.} function of a compact support being the closure of $\bigcup_\beta W^{{j}a}_\beta$:
\begin{equation}
\supp \varphi^{{j}a}_{\sigma}=\overline{\Big(\bigcup_\beta W^{{j}a}_\beta\Big)}.
\label{supp}
\end{equation} 
These properties of $\varphi^{{j}a}_{\sigma}$ guarantee that for every $\lambda\in[\lambda_-,\lambda_+]$ and for every $\sigma\in[1,\infty[$ 
\[
|\varphi^{{j}a}_{\sigma}\theta^1(\dot{\tilde{f}}_{{j}a})|\leq |\theta^1(\dot{\tilde{f}}_{{j}a})|.
\]
Because the latter function is integrable on $[\lambda_-,\lambda_+]$ the  Lebesgue's dominated convergence theorem allows us to conclude that the following function 
\[
\sigma \mapsto \Phi^{{j}a}(\sigma):=\int_{g_{{j}a}} \varphi^{{j}a}_{\sigma}\theta^1=\int_{\lambda_-}^{\lambda_+} \varphi^{{j}a}_{\sigma}\theta^1(\dot{\tilde{f}}_{{j}a})\,d\lambda
\] 
is continuous on $[1,\infty[$. If $\sigma=1$ then $\varphi^{{j}a}_\sigma=\phi^{{j}a}$ and (see \eqref{2/5})
\[
|\Phi^{{j}a}(1)|<\frac{2}{5}|x^0_{{j}}|. 
\]  
Moreover, {the properties of $\phi^{ja}$ guarantee that} if the parameter $\sigma$ goes to the infinity then the family $\{\varphi^{{j}a}_{\sigma}\}$ converges pointwisely to the characteristic function on $\bigcup_\beta W^{{j}a}_\beta$. Thus by virtue of the Lebesgue's theorem and \eqref{3/5} 
\[
\lim_{\sigma\to\infty} \Phi^{{j}a}(\sigma)=(-1)^a\frac{3}{5}x^0_{{j}}.  
\]
All these mean that for each $\Phi^{{j}a}$ there exist $\sigma\in[1,\infty[$ such that
\begin{equation}
\Phi^{{j}a}(\sigma)=\int_{g_{{j}a}} \varphi^{{j}a}_{\sigma}\theta^1=(-1)^a\frac{x^0_{{j}}}{2}.
\label{Phi-Ji}
\end{equation}
Denote the function  $\varphi^{{j}a}_{\sigma}$ with this special $\sigma$ by $\varphi^{{j}a}$.

In this way we defined functions $\{\varphi^{{j}a}\}$ for indices $\{{j}\}$ such that $x^0_{{j}}\neq 0$. If $x^0_{{j}}=0$ then set
\begin{equation}
\varphi^{ja}=0
\label{varphi=0}
\end{equation}
everywhere on $\Sigma$.
 
Now to finish the proof it is enough to define
\[
\alpha_1:=\sum_{{j}=1}^N\sum_{a=1}^2 (-1)^a\varphi^{{j}a}.
\]
Indeed, by virtue of \eqref{supports} supports of the functions $\{\varphi^{{j}a}\}$ are pairwise disjoint. This fact, \eqref{varphi<1} and \eqref{varphi=0} guarantee that $\alpha_1$ satisfies \eqref{al<1}. Moreover, because of \eqref{g-W}, {\eqref{supp}} and \eqref{varphi=0}
\[
\gamma\cap \supp \varphi^{{j}a}=
\begin{cases}
g_{{j}a}\subset \intr f_{{j}a} & \text{if $x^0_{{j}}\neq 0$,}\\
\varnothing & \text{otherwise}
\end{cases}.
\]
Taking into account \eqref{Phi-Ji} we conclude that Equations \eqref{th-0-an} are satisfied.



\end{document}